%% file: main.tex
\documentclass[lettersize,journal]{IEEEtran}
\usepackage{amsmath,amsfonts}
\usepackage{algorithmic}
\usepackage{algorithm}
\usepackage{amsmath,amssymb,amsthm}
\usepackage{array}
\usepackage[caption=false,font=normalsize,labelfont=sf,textfont=sf]{subfig}
\usepackage{textcomp}
\usepackage{stfloats}
\usepackage{url}
\usepackage{verbatim}
\usepackage{graphicx}
\usepackage{cite}
\usepackage{soul}
\usepackage{xcolor}
\newtheorem{definition}{Definition}

\usepackage{amsthm}
\newtheorem{lemma}{Lemma}
\newtheorem{corollary}{Corollary}

\usepackage{booktabs}

\usepackage{hyperref}
\usepackage{siunitx}
\usepackage[acronym]{glossaries}
\input{acronyms}

\begin{document}


\title{Modality-Tailored Age of Information for Multimodal Data in Edge Computing Systems}

\author{
Ying Liu, Yifan Zhang,~\IEEEmembership{Graduate Student Member,~IEEE}, Xinyu Wang,~\IEEEmembership{Graduate Student Member,~IEEE},\\ Chao Yang,~\IEEEmembership{Graduate Student Member,~IEEE}, Kandaraj Piamrat,~\IEEEmembership{Member,~IEEE}, \\Stephan Sigg,~\IEEEmembership{Senior Member,~IEEE}, Zheng Chang,~\IEEEmembership{Senior Member,~IEEE}, Yusheng Ji,~\IEEEmembership{Fellow,~IEEE}

\thanks{Ying Liu, Yifan Zhang, and Stephan Sigg are with the Department of Information and Communications Engineering, Aalto University, 02150 Espoo, Finland (e-mail: ying.2.liu@aalto.fi, yifan.1.zhang@aalto.fi, stephan.sigg@aalto.fi)}
\thanks{Xinyu Wang is with the School of Computer Science and Technology, University of Science and Technology of China, Hefei 230027, China (e-mail: xinyuwang0306@mail.ustc.edu.cn).}
\thanks{Chao Yang is with the Department of Energy and Mechanical Engineering, Aalto University, 02150 Espoo, Finland (e-mail: chao.1.yang@aalto.fi).}
\thanks{Kandaraj Piamrat is with the Nantes University, École Centrale Nantes, CNRS, INRIA, LS2N, UMR 6004, 44000 Nantes, France (e-mail: kandaraj.piamrat@ls2n.fr).}
\thanks{Zheng Chang is with the Faculty of Information Technology, University of Jyv\"askyl\"a, 40014 Jyv\"askyl\"a, Finland (e-mail: zheng.chang@jyu.fi)}
\thanks{Yusheng Ji is with the Information Systems Architecture Science Research Division, National Institute of Informatics, Tokyo 101-0003, Japan (e-mail:kei@nii.ac.jp).}}

\markboth{Journal of \LaTeX\ Class Files,~Vol.~14, No.~8, August~2021}%
{Shell \MakeLowercase{\textit{et al.}}: A Sample Article Using IEEEtran.cls for IEEE Journals}


\maketitle

\begin{abstract}
As Internet of Things (IoT) systems scale and device heterogeneity grows, multimodal data have become ubiquitous.
Meanwhile, evaluating the freshness of multimodal data is essential, as stale updates would delay task execution, degrade decision accuracy, and undermine safety in latency-sensitive services.
However, existing freshness metrics such as Age of Information (AoI) are not suitable for multimodal data, as they do not capture modality-specific characteristics.
In this paper, we propose a metric, namely, Modality-Tailored Age of Information (MAoI), to provide a unified and decision-relevant evaluation of freshness for resource management and policy optimization for multimodal data.
This metric integrates modality-specific semantic and temporal characteristics, reflecting both age evolution and content importance for multimodal data in multi-access edge computing (MEC) systems.
Then, the closed-form expression of the average MAoI is derived, and an MAoI minimization problem is formulated, where sampling intervals and offloading decisions are optimized with practical energy constraints.
To effectively solve this problem, a Joint Sampling Offloading Optimization (JSO) algorithm is proposed to jointly optimize the sampling intervals and offloading decisions. It is a block coordinate descent-based algorithm where an optimal sampling-interval subalgorithm is used to update the sampling intervals, and an interference-aware best-response offloading subalgorithm is proposed to update the offloading decisions alternately.
Finally, a comprehensive simulation is performed, confirming that the MAoI metric effectively quantifies multimodal freshness compared to traditional AoI, and the JSO algorithm significantly minimizes the average MAoI compared to state-of-the-art algorithms.
\end{abstract}






\section{Introduction}\label{section1}
\IEEEPARstart{M}{ultimodal} perception has become a ubiquitous capability in Internet of Things (IoT) systems~\cite{He2022Collaborative, fu2023survey, wan2022ai}.
In these systems, heterogeneous data such as images, audio, radar signals, and other sensing modalities are jointly exploited to provide complementary, more reliable, and more diverse views~\cite{chen2024lightweight, liu2023mmfusion, ferranti2023survey}. This supports many practical applications, such as traffic monitoring~\cite{cheng2024tits}, pedestrian detection~\cite{tian2024rgbt}, and smart healthcare~\cite{alsamhi2022edge}. 
 These deployments collectively highlight the growing importance of heterogeneous multimodal systems in modern IoT~\cite{Mondal2025Survey}.
Meanwhile, in such systems, information freshness is particularly critical because outdated information can delay task execution~\cite{Shiraishi2023WCL}, reduce decision accuracy~\cite{mena2023automatica_aoi_ncs}, and weaken safety guarantees~\cite{zhang2020tits_cacc_delay} in latency-sensitive applications such as cooperative driving~\cite{wang2025latencycp}, industrial control~\cite{riihinen2024iiot_realtime}, and healthcare monitoring~\cite{habibzadeh2020hiot}.

Therefore, how to evaluate freshness in multimodal settings becomes critical. 
Most existing schemes use the Age of Information (AoI)~\cite{kaul2012real} or its task-aware variants~\cite{chiariotti2022qao, li2021iotj_aop, jayanth2022twc_aopi} as freshness metrics in similar systems. AoI measures the time elapsed since the most recently generated update was received, which ties freshness to system-wide update dynamics~\cite{kaul2012real}. 
Its variants change where and how age is measured to match different objectives, such as Query AoI~\cite{chiariotti2022qao}, the Age of Processing~\cite{li2021iotj_aop}, and the Age of Processed Information~\cite{jayanth2022twc_aopi}. 

With these metrics, prior work has built age-aware scheduling in MEC~\cite{he2024tmc}, tuned sampling, uplink, and offloading under device and link limits~\cite{jiang2023iotj}, traded energy for freshness on battery-powered devices~\cite{liu2022wcl}, synchronized digital twins and set triggers at the edge~\cite{Li2024DTContinual,li2024ton,zhang2024tsc}, and analyzed how shared queues in computing and transmission affect freshness~\cite{tcom2024_aoi_edgequeue}. 
However, the AoI and its variants are modality-agnostic. Consequently, these metrics do not account for the distinct sampling, queuing, and inference pipelines across multimodal data such as images~\cite{he2016resnet}, audio~\cite {Chan2021LAS}, and other signals~\cite{foumani2024dltsc, zhang2023mmwave}. They also cannot capture the different timeliness priorities of multimodal data. For example, real-time video streams~\cite{xu2023evacomst} can often tolerate lower freshness than safety-critical biosignals~\cite{charlton2023ppg}.
Hence, AoI and its variants are unable to capture the key characteristics of multimodal data.
As a result, using modality-agnostic age metrics in multimodal systems would lead to suboptimal update prioritization and inefficient resource allocation. 

Actually, designing a freshness metric for multimodal data is both urgent and technically challenging.
To reflect the temporal recency of the multimodal data, the metric requires increasing monotonically with staleness. 
The metric should also remain sensitive to task-relevant content changes to capture modality-specific semantics and temporal dynamics. 
Meanwhile, cross-modal comparability is required to weigh the freshness gain of one modality against another on a common scale, thereby optimizing system resource management. 
The metric also needs to be derivable from low-overhead features to support fast online control without incurring substantial time or resource costs. 
These requirements pose a challenging design problem for multimodal freshness metrics.

To address these challenges, we propose the Modality-Tailored Age of Information (MAoI) to evaluate the freshness of the multimodal data. 
MAoI captures multimodal heterogeneity by modeling modality-dependent sampling processes and by accounting for inference latency through representative neural pipelines aligned with each modality’s characteristics. 
In addition, MAoI maps modality-specific features to age growth to integrate the semantic and temporal features of multimodal data and to reflect the different priorities across different modalities. These modality-specific features are low-overhead, so the metric can be computed online and used for fast control without imposing significant time or resource costs.

Then, the MAoI is instantiated in a multi-access edge computing (MEC) system with numerous IoT devices and an MEC server co-located with a base station (BS).
Building on these, we minimize the average MAoI by optimizing the sampling intervals and offloading decisions under practical device energy constraints.
Specifically, a block coordinate descent (BCD) based algorithm named Joint Sampling Offloading Optimization (JSO) is developed, which decouples the mixed continuous–discrete problem into a sampling block and an offloading block.
In the sampling block, an optimal sampling-interval subalgorithm constructs a convex upper-bound surrogate of the MAoI objective to compute a closed-form per-device update, followed by a projected-Newton refinement to enforce bounds and improve accuracy. In the offloading block, an interference-aware best-response offloading subalgorithm models the offloading decisions as a distributed game under uplink interference, where each device updates its action via an interference-aware best response strategy given the current decisions of other devices. These two subalgorithms alternate within the BCD scheme to update their respective blocks, yielding an effective joint policy for sampling and offloading. Finally, comprehensive simulations are conducted to validate the effectiveness of the MAoI and the proposed optimization algorithm. 
The main contributions are summarized as follows
\begin{enumerate}
    \item MAoI is proposed for the freshness evaluation of the multimodal data, which integrates modality-specific semantic and temporal characteristics, reflecting both age evolution and content importance. Specifically, image MAoI grows with content dynamics and the region-of-interest ratio within the frame; audio MAoI grows with semantic variation and perceptual-quality thresholds; and signal MAoI grows with temporal dynamics and measurement quality.
    
    \item A MAoI minimization algorithm is developed, where sampling intervals and offloading decisions are jointly optimized under the constraints of device energy. Specifically, the JSO algorithm is used to solve the problem efficiently, in which an optimal sampling-interval subalgorithm is used to update the sampling intervals, and an interference-aware best-response offloading subalgorithm is proposed to update the offloading decisions.
    
    \item Extensive simulations are conducted to validate the effectiveness of MAoI on multimodal workloads, showing that MAoI better quantifies multimodal freshness than traditional AoI and that the proposed JSO algorithm consistently lowers the system-level average MAoI while meeting constraints compared to state-of-the-art algorithms.
\end{enumerate} 

The remainder of this paper is organized as follows. Section~\ref{sec:system_model} introduces the system model, covering modality-specific time model, the offloading model, and the energy consumption model. Section~\ref{sec:maoi_proposal} defines MAoI, proposes per-modality MAoI and derives the closed-form average MAoI expression, and formulates the average MAoI minimization problem. Section~\ref{sec:algorithm} presents the JSO algorithm. Section~\ref{sec:experimental} reports simulations that validate MAoI compared with AoI and evaluate the JSO algorithm. Section~\ref{sec:conclusion} concludes this work.

\begin{figure*}
    \centering
\includegraphics[width=0.9\linewidth]{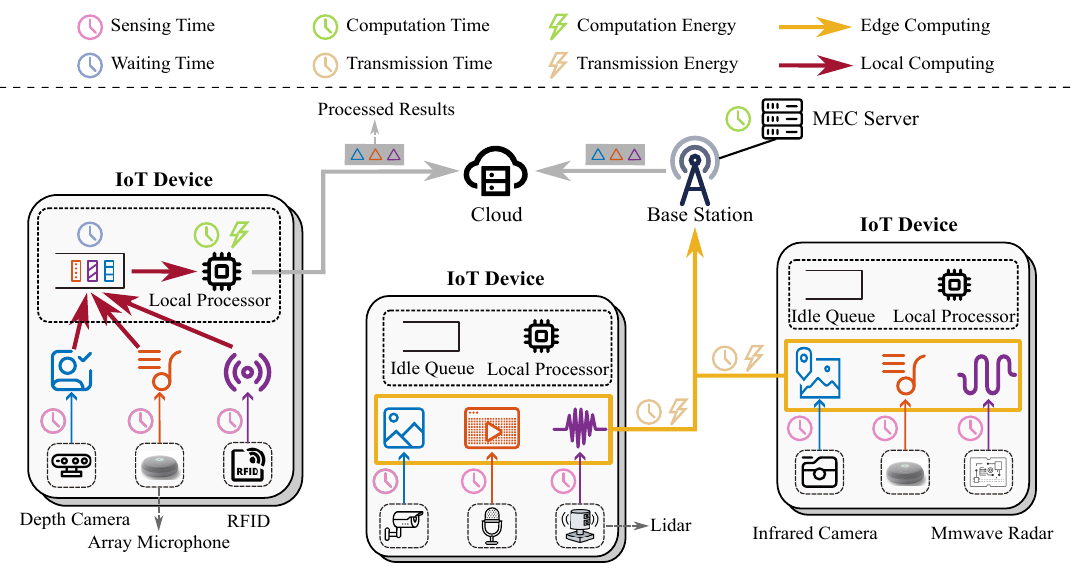}
    \caption{System Model. IoT devices with heterogeneous sensors generate multimodal updates, including image, audio, and other signal modalities. These updates are either processed locally or offloaded to a BS and executed at the MEC server. 
    Processed results from the local processor or the MEC server are forwarded to the cloud.
    For local computing, the system time consists of sensing time, waiting time due to limited computing resources, and local computation time, with energy consumption mainly from computation. If the computation updates are offloaded to the MEC server, the system time consists of sensing time, transmission time, and edge computation time (in parallel across modalities), with energy consumption mainly from transmission.}
    \label{fig:system_model}
\end{figure*}

\section{System Model}\label{sec:system_model}
As depicted in Fig.~\ref{fig:system_model}, a MEC system that consists of a set of IoT devices and a MEC server co-located with a BS is considered.
The IoT devices, represented by a set $\mathcal{D}$ with cardinality $D$, are equipped with a local processor and various sensors (e.g., cameras, audio sensors, temperature sensors, and radar), forming a set $\mathcal{S}$ with cardinality $S$.
These sensors continuously monitor the physical environment and generate multimodal status updates.
Conditioned on the limited computational capability of the local processor, each IoT device either processes its sensed data locally or offloads computation tasks over the device-to-BS uplink for execution at the MEC server. 


To maintain the freshness of status information and ensure effective system control, the entire sensing, offloading, and processing cycle is executed repeatedly. IoT devices follow the \textit{generate-at-will} sampling policy~\cite{yates2015lazy}, under which each device initiates a new status update only after completing the previous one, thereby preventing unnecessary waiting.
Each status update from an IoT device $d$, denoted by $U_{d} = \{ U_{d,s} \mid s = 1,2,3 \}$, consists of three representative modalities: image ($s=1$), audio ($s=2$), and other signals ($s=3$). After generation, the device determines whether to process these data locally or offload them to the MEC server.

\subsection{Modality-Specific Sensing and Computation Time Models}
In this section, modeling approaches for sensing and computation associated with image, audio, and other signal modalities are discussed. 
The sensing time is modeled based on the distinct data acquisition characteristics specific to each modality. Regarding computation time, given the increasing significance and extensive adoption of deep learning techniques in IoT data processing, the computation time model is modeled using the computational complexity of modality-specific deep learning algorithms, rather than relying solely on data size as in traditional methods.

\subsubsection{Image Modality}
\paragraph{Sensing Time Model} The sensing time required by image sensors, such as cameras, is typically negligible compared to transmission and computation times. Therefore, the image sensing time for device $d$ is modeled as
\begin{IEEEeqnarray}{c}
    T_{d,1}^{\text{sens}} = 0.
\end{IEEEeqnarray}

\paragraph{Computation Time Model} 
It is assumed that each device $d$ captures images of resolution $H_d\times W_d$ (height × width) in RGB format, consisting of 3 color channels with 8 bits per channel. 
Thus, the data size of image data generated by device $d$ is calculated as
\begin{equation}
S_{d}^{\text{img}} = H_d \times W_d \times 3 \times 8.
\end{equation}

An adaptive ResNet~\cite{he2016resnet} is adopted with residual connections to process the image data. Its classification head uses adaptive average pooling, so the network accepts variable-size inputs while producing a fixed-length representation. The network's computational cost is dominated by convolutions and scales approximately with spatial area. Let $O_{\text{ResNet}}^{\text{base}}$ denote the complexity of the same architecture for a $224\times224$ reference input. Accordingly, we model the complexity as a function of the input resolution $(H_d, W_d)$ and scale it from the baseline, which is
\begin{equation}
O_{\text{ResNet}}(W_d,H_d) \approx O_{\text{ResNet}}^{\text{base}} \times \frac{W_d\times H_d}{224\times224}.
\end{equation}


Consequently, the computation time for an image generated by device $d$ is calculated as
\begin{equation}
T_{d,1}^{c,\text{comp}} = \frac{O_{\text{ResNet}}(W_d,H_d)}{f_{c}}, \quad c \in \{\text{local},\text{edge}\},
\end{equation}
where \( f_{c} \) represents the computational capability at the computing location \( c \). Specifically, when \( c = \text{local} \), the computation is performed locally by the IoT device; when \( c = \text{edge} \), the computation is executed at the MEC server.

\subsubsection{Audio Modality}
\paragraph{Sensing Time Model} The sensing duration for audio data directly corresponds to the recorded audio segment duration $t_{d}^{\text{aud}}$. Thus, the audio sensing time is
\begin{equation}
T_{d,2}^{\text{sens}}=t_{d}^{\text{aud}}.
\end{equation}

\paragraph{Computation Time Model} Audio data from device $d$ is characterized by duration $t_{d}^{\text{aud}}$, sampling rate $R_{d}^{\text{aud}}$, bit depth $b_{d}^{\text{aud}}$, and number of audio channels $n_{d}^{\text{aud}}$. Hence, the audio data size is expressed as
\begin{equation}
S_{d}^{\text{aud}}=t_{d}^{\text{aud}}\times R_{d}^{\text{aud}}\times n_{d}^{\text{aud}}\times b_{d}^{\text{aud}}.
\end{equation}

For audio data processing, DeepSpeech2~\cite{amodei2016deep} is utilized, which is an advanced model composed of convolutional layers followed by recurrent neural network layers with bidirectional long short-term memory (BiLSTM). 
Due to the sequential processing nature of BiLSTM layers and the convolutional layers’ linear scaling property with respect to input length, the computational complexity of DeepSpeech2
is effectively proportional to the audio duration, which can be expressed as
\begin{equation}
O_{\text{DS2}}(t_{d}^{\text{aud}}) \approx O_{\text{DS2}}^{\text{base}} \times t_{d}^{\text{aud}},
\end{equation}
where $O_{\text{DS2}}^{\text{base}}$ is the baseline per-second computational cost of the chosen DeepSpeech2 inference model.
Thus,  the computation time for audio modality generated by device $d$ is
\begin{equation}
T_{d,2}^{c,\text{comp}} = \frac{O_{\text{DS2}}(t_{d}^{\text{aud}})}{f_{c}}, \quad c \in \{\text{local},\text{edge}\}.
\end{equation}

\subsubsection{Other Signal Modality}
\paragraph{Sensing Time Model}
Other signal data, such as environmental, physiological, or radar measurements, are continuously sampled and segmented into frames for downstream processing.
The sensing time for other signal data is defined as the total observation duration $t_d^{\text{sig}}$, which is given by
\begin{equation}
        T_{d, 3}^{\text{sens}} = t_d^{\text{sig}}.
\end{equation}
Given the frame rate $f_{d}^{\text{frame}}$ (frames per second), the total number of frames $L_d$ for other signal data is calculated as
\begin{equation}
        L_d = f_{d}^{\text{frame}} \times t_d^{\text{sig}}.
\end{equation}
For each frame, feature extraction or detection algorithms yield $N_{d}^{\text{sig}}$ detected points. Each point is described by $F_d^{sig}$ features (such as spatial coordinates, amplitude, or velocity). 
Thus, the amount of data to be transmitted or processed for the entire signal segment is given by
\begin{equation}
        S_{d}^{\text{sig}} = L_d \times N_{d}^{\text{sig}} \times F_d^{sig} \times w_d^{\text{sig}},
\end{equation}
where $w_d^{\text{sig}}$ denotes the number of bits used to represent each feature.

\paragraph{Computation Time Model}
To model complex temporal dependencies in other signal data, the Temporal Fusion Transformer (TFT)~\cite{lim2021temporal} is employed. The self-attention mechanism within TFT requires pairwise interactions across all frames of the input sequence, resulting in an overall computational complexity that scales quadratically with the number of frames, which can be expressed as
\begin{align}
    O_{\text{TFT}}(L_{ d}) &\approx O^{\text{base}}_{\text{TFT}} \cdot \frac{L_{ d}^2}{L_{\text{base}}^2},
\end{align}
where $O^{\text{base}}_{\text{TFT}}$ denotes the computational complexity for a reference sequence length $L_{\text{base}}$ of the TFT networks.
Accordingly, the computation time for other signal data is
\begin{equation}
    T^{c, \text{comp}}_{d, 3} = \frac{O_{\text{TFT}}(L_{ d})}{f_c}, \quad c \in \{\text{local}, \text{edge}\}.
\end{equation}

\subsection{Offloading Model}
In the proposed system, IoT devices can process their multimodal status updates either locally or offload them to the MEC server for processing.
A device-level binary offloading policy is adopted: for each device $d$, $x_d\!\in\!\{0,1\}$ indicates whether its modalities are offloaded ($x_d\!=\!1$) or processed locally ($x_d\!=\!0$).
Let $\mathbf{x}=\{x_{d}\mid d=1,2,\dots,D\}$ denote the offloading decisions of all devices, with
\begin{IEEEeqnarray}{c}
     x_{d} =
\begin{cases}
1,& \text{offloaded to the MEC server,}\\
0,& \text{processed locally.}
\end{cases}
\end{IEEEeqnarray}

\subsubsection{Local Computing}
Due to limited computational capabilities, each IoT device processes the different modalities (image, audio, other signal) sequentially, introducing waiting times dependent on the processing order. Thus, the system time for local computing for modality $s$ of device $d$ is
\begin{equation}
T_{d,s}^{\text{local,sys}}=T_{d,s}^{\text{sens}}+W_{d,s}+T_{d,s}^{\text{local,comp}},\; s\in\{1, 2, 3\},
\end{equation}
where $W_{d,s}$ denotes the waiting time. 
Specifically, the waiting time $W_{d,s}$ for modality $s$ at device $d$ is equal to the total local computation time of all modalities scheduled ahead of it. Thus, we have
\begin{equation}
W_{d,s} = \sum_{s' \in \mathcal{S}_d^{\text{before}}(s)} T_{d,s'}^{\text{local,comp}}, \; s\in\{\text{1, 2, 3}\},
\end{equation}
where $\mathcal{S}_d^{\text{before}}(s)$ denotes the set of modalities processed before modality $s$ at device $d$.

\subsubsection{Edge Computing}
When offloading is adopted, the multimodal status updates are first forwarded to the BS and then executed at the co-located MEC server.
As all modalities of a device are offloaded jointly, the total uplink transmission time of device $d$ is
\begin{equation}
T_{d}^{\text{trans}}(\mathbf{x})= \frac{S_d}{r_d(\mathbf{x})} =\frac{S_{d}^{\text{img}} + S_{d}^{\text{aud}} + S_{d}^{\text{sig}}}{r_{d}(\mathbf{x})},
\label{eq:transmission_rate}
\end{equation}
where $r_{d}(\mathbf{x})$ denotes the transmission rate given by
\begin{IEEEeqnarray}{c}
    r_{d}(\mathbf{x})= B\log_2\left(1+\frac{P_{d}g_{d}^{\text{edge}}}{\omega_0+ {\textstyle \sum_{j=1, j \neq d}^{D}}x_{j}P_{j}g_{j}^{\text{edge}} } \right),
    \label{eq:rate_definition}
\end{IEEEeqnarray} 
{where $B$ is the uplink bandwidth, $P_d$ and $P_j$ are the transmit powers of devices $d$ and $j$, $g_{d}^{\text{edge}}$ and $g_{j}^{\text{edge}}$ denote the device–to-BS channel gains, the term $\sum_{j\neq d} x_j P_j g_{j}^{\text{edge}}$ collects the co-channel uplink interference aggregated at the BS, and $\omega_0$ is the noise power.}

Given the richer computational resources at the MEC server, all modalities can be processed simultaneously, thereby eliminating any waiting time. Hence, the total system time for edge computing is given by
\begin{equation}
        T_{d,s}^{\text{edge,sys}}(\mathbf{x})\!= \!T_{{d,s}}^{\text{sens}} \!+ \!T_{d}^{\text{trans}}(\mathbf{x})\!+\! T_{d,s}^{\text{edge,comp}}, s\in\{{1, 2, 3}\}.
\end{equation}

Overall, the system time is calculated as
\begin{equation}
    T_{d,s}^{\text{sys}}(\mathbf{x})\!= \!x_{d}T_{d,s}^{\text{edge,sys}}(\mathbf{x}) \!+\!(1-x_{d}) T_{d,s}^{\text{local,sys}}, \; s\in\{{1, 2, 3}\}.
    \label{eq:system_time}
\end{equation} 

\subsection{Energy Consumption Model}
This subsection characterizes the energy consumption of each IoT device, consisting of sensing, computation, and transmission energy. 

\subsubsection{Sensing Energy}
The per-update sensing energy is composed of three parts, corresponding to the image, audio, and other signal modalities.

For the image data, a sensing energy model grounded in recent CMOS image sensor (CIS) studies~\cite{ma2023camj} is used, which attributes the per-capture cost mainly to pixel readout, ADC conversion, and on-sensor I/O. Accordingly, the sensing energy model is expressed as
\begin{equation}
E_{d}^{\text{img}}=\kappa_d^{\text{cam}}+e_d^{\text{pix}}\,(H_d\,W_d\, c^{\text{img}}),
\label{eq:Eimg}
\end{equation}
where $\kappa_d^{\text{cam}}$ denotes frame-level overhead energy (e.g., sensor wake-up and initial on-sensor preprocessing), $e_d^{\text{pix}}$ is the effective per-pixel energy, and $c^{\text{img}}$ denotes the number of channels for the image data.

For the audio modality, analog-to-digital converter (ADC) figures-of-merit and system measurements indicate that, under a fixed sampling configuration, the microphone analog front-end together with the ADC draws approximately constant power. Hence, the per-update energy grows linearly with recording duration~\cite{tang2022saradc, cao2023powerphone}, which is modeled as
\begin{equation}
E_{d}^{\text{aud}}=\big(P_{d}^{\text{aud}}+\alpha_{d}\,R_{d}^{\text{aud}}\,b^{\text{aud}}_{d}\,n^{\text{aud}}_{d}\big)\,t_{d}^{\text{aud}},
\label{eq:Eaud}
\end{equation}
where $P_{d}^{\text{aud}}$ is the configuration-dependent baseline power, $\alpha_d$ is a proportionality constant capturing ADC/I/O scaling.

For other signal modalities, millimeter-wave radar serves as a representative case. Vendor guides and datasheets show that, with a fixed chirp/frame configuration, the active sensing power is approximately constant. Consequently, the sensing energy over one update is proportional to the active-on time~\cite{tiAWR6843AOP}, which is given by
\begin{equation}
E_{d}^{\text{sig}}=P_{d}^{\text{sig}}\,t_{d}^{\text{sig}},
\label{eq:Esig}
\end{equation}
where $P_{d}^{\text{sig}}$ is the average active sensing power under the chosen radar configuration.

Accordingly, the per-update sensing energy for device $d$ is given by
\begin{equation}
E_{d}^{\text{sens}} = E_{d}^{\text{img}} + E_{d}^{\text{aud}} + E_{d}^{\text{sig}}.
\label{eq:Esens_total}
\end{equation}

\subsubsection{Computation Energy}
If device $d$ performs local computing ($x_{d}=0$), its computation energy is modeled as linearly proportional to the algorithmic complexity of the processing algorithms for each modality. Thus, we have
\begin{equation}
E_{d}^{\text{comp}}\!=\!\gamma\left[O_{\text{ResNet}}(W_d,H_d)\!+\!O_{\text{DS2}}(t_{d}^{\text{aud}})\!+\!O_{\text{TFT}}(L_d)\right],
\end{equation}
where $\gamma$ denotes the energy consumption coefficient.

\subsubsection{Transmission Energy}
If device $d$ offloads its status updates ($x_{d}=1$), instead of computational energy, transmission energy is required
\begin{equation}
E_{d}^{\text{trans}}(\mathbf{x})=P_{d}T_{d}^{\text{trans}}(\mathbf{x}),
\end{equation}
where $P_{d}$ denotes the transmission power of device $d$.

Overall, the total energy consumption of the device $d$ is
\begin{equation}
E_{d}(\mathbf{x})=E_{d}^{\text{sens}} + (1 - x_{d})E_{d}^{\text{comp}}+x_{d}E_{d}^{\text{trans}}(\mathbf{x}).
\end{equation}

\section{MAoI Proposal and Problem Formulation}\label{sec:maoi_proposal}
This section defines the MAoI for the image, audio, and signal modalities. The closed-form expressions for the average MAoI of each modality are then derived. Finally, a system-level optimization problem that minimizes the average MAoI is formulated.

\subsection{MAoI Proposal for the Image Modality}
We propose an MAoI metric specifically for image modality, explicitly incorporating two key attributes: the dynamism of image content and the proportion of important regions within the image.

\paragraph{Dynamism of Image Content}
Images exhibit varying levels of dynamism; static images (e.g., landscape photos) generally convey less time-critical information, whereas dynamic images (e.g., surveillance footage or video frames) often contain rapidly changing, critical information that necessitates timely updates. We quantify this image dynamism $I_{d}^{\text{img}}$ using the average absolute pixel-wise difference between consecutive images
\begin{IEEEeqnarray}{c}
I_{d}^{\text{img}} = \frac{1}{M} \sum_{m=1}^{M} \left| p_{d, i}^m - p_{d, i-1}^m \right|,
\end{IEEEeqnarray} 
where $p_{d, i}^m$ denotes the value of the $m$-th pixel in the $i$-th image, and $M$ is the total number of pixels.

\paragraph{Proportion of Key Regions}
Certain images contain regions of particular interest, such as faces or key objects, that carry richer and more critical information requiring prompt processing. These regions can be automatically identified using deep learning-based object detection techniques. We define $R_{d}^{\text{img}}$ as the proportion of these key regions relative to the total image area
\begin{IEEEeqnarray}{c}
R_{d}^{\text{img}} = \frac{A_d^{\text{roi}}}{A_d^{\text{total}}},
\end{IEEEeqnarray} 
where $A_{\text{roi}}$ is the area occupied by the regions of interest, and $A_{\text{total}}$ is the total image area.

\paragraph{Event-Triggered Growth and Image Average MAoI}
In the proposed system, each device generates status updates at a constant sampling interval $\tau_d$. 
We define {image events} as noticeable changes in the captured image content within a sampling interval $\tau_d$. 
Such events typically occur when the overall scene dynamics increase rapidly or when the proportion of key regions varies markedly, making the image more informative for decision-making. 
For analytical tractability, the occurrence of image events is modeled as a Poisson process with average generation rate $\lambda_1$. 

When an image event occurs within the interval $\tau_d$, the growth rate $K_{d,1}$ of the MAoI for the image modality is enlarged by incorporating the contributions from image dynamism and key region proportion; otherwise, it remains at the baseline level. Formally, 
\begin{IEEEeqnarray}{c}
K_{d,1} =
\begin{cases}
1, & \text{if no event occurs},\\
1 +I_{d}^{\text{img}} + R_{d}^{\text{img}},&\text{otherwise}.
\end{cases}
\end{IEEEeqnarray}
This formulation enables the average image MAoI to reflect both the frequency of informative content changes and their relative importance.
Thus, the probability mass function (PMF) of $K_{d,1}$ is
\begin{IEEEeqnarray}{c}
\mathbb{P}(K_{d,1}) =
\begin{cases}
1 - e^{-\lambda_1 \tau_d}, & K_{d,1} = 1 + I_{d}^{\text{img}} + R_{d}^{\text{img}},\\
e^{-\lambda_1 \tau_d}, & K_{d,1} = 1.
\end{cases}
\label{pmf:k1}
\end{IEEEeqnarray}

The MAoI for image data at time $t$ is defined as
\begin{IEEEeqnarray}{c}
\Delta_{d,1}(t) = k_{d,1}^{(i)} \big(t - a_d(t)\big), \quad t \in (t_i, t_{i+1}],
\end{IEEEeqnarray}
{where }$\{t_i\}_{i\ge 0}$ are the generation time of successive image updates at device $d$, $a_d(t)=t_i\ \text{on}(t_i,t_{i+1}]$,\ {and } $k_{d,1}^{(i)}$ is the realization of $K_{d,1}${over that interval.}

To calculate the average MAoI for the image modality, we consider an example of device $d$. Fig.~\ref{fig:AoI_evo} depicts the evolution of the MAoI for the image modality in the proposed multimodal system. Without loss of generality, we assume the initial observation time is $t_0 = 0$, and the initial MAoI is denoted as $\Delta_{d,1}(0)$.
\begin{figure}
    \centering
\includegraphics{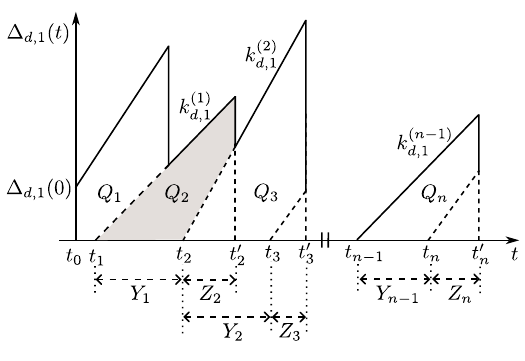}
    \caption{An example of the evolution of MAoI for the image modality.}
    \label{fig:AoI_evo}
\end{figure}
The MAoI for the $i$-th image update escalates at a rate defined by the slope $k_{d,1}^{(i)}$, continuing until the subsequent image update is sent to the cloud. When the system goes to a steady state, $\{k_{d,1}^{(i)}\}$ are independent and identically distributed (i.i.d.) realizations of $K_{d,1}$.

Let $Y_{i}$ denote the interval between the $i$-th and $(i+1)$-th image data generation. 
 \begin{IEEEeqnarray}{c}
     Y_i = t_{i+1} - t_{i}, 
 \end{IEEEeqnarray} 
where $t_i$ is the generation time of the $i$-th image updates. 
Denote $Z_{i}$ as the system time of the $i$-th image update
 \begin{IEEEeqnarray}{c}
     Z_{i} = t_{i}' - t_{i},
 \end{IEEEeqnarray} 
where $t_{i}'$ represents the service completion time of the $i$-th image update. Then, in a time period $[0,T]$, the average MAoI for the image modality of device $d$ is
\begin{IEEEeqnarray}{c}
\Delta _{d,1}(T)=\frac{1}{T}{\int_0^T {\Delta}_{d,1}(t) \, dt}.
\end{IEEEeqnarray} 
For the ease of exposition, we set $T=t_{n}'$. As shown in Fig.~\ref{fig:AoI_evo}, the area under the MAoI curve can be calculated as a concatenation of polygonal areas $Q_i$ and the triangular area within the interval $\left(t_{n}, t_{n}'\right)$. Therefore, the average MAoI for the image modality is expressed as
\begin{IEEEeqnarray}{c}
\Delta_{d,1}(T) = \frac{1}{T} \left(\sum_{i=1}^{n} Q_{i}+\frac{1}{2}\left({t_{n}'-t_{n}}\right)^{2}\right).
\end{IEEEeqnarray} 
Let $ C=\frac{1}{2}{(t_{n}'-t_n)^2 + Q_1} $, we can obtain
\begin{IEEEeqnarray}{c}
\Delta_{d,1}(T) =\frac{C}{T}+\frac{n-1}{T} \frac{1}{n-1}\sum _{i=2}^n Q_i.
\label{5}
\end{IEEEeqnarray} 
When $T$ goes to infinity, the term ${C}/{T}$ converges to zero as $C$ is finite and ${n-1}/{T}$ can be regarded as the value of ${1}/{\mathbb{E}{(Y_i)}}$. 

Letting {$T$} go to infinity, the average MAoI for the image modality can be expressed as
\begin{IEEEeqnarray}{c}
\begin{aligned}
{\bar\Delta}_{d,1} &= \lim_{T \to \infty}\Delta_{d,1}(T) = \frac{1}{\mathbb{E}\left(Y_i\right)}\left(\frac{1}{n-1}\sum _{i=2}^n Q_i\right)\\&=\frac{\mathbb{E}\left(Q_i\right)}{\mathbb{
E}\left(Y_i\right)},
\end{aligned}
\end{IEEEeqnarray} 
where $\mathbb{E}\left(\cdot\right) $ is the expectation. 
As $Q_i$ can be computed as
\begin{IEEEeqnarray}{c}
\begin{aligned}
Q_i &= \frac{1}{2}k_{d,1}^{(i-1)}(Y_{i-1} + Z_i)^2 - \frac{1}{2}k_{d,1}^{(i)} Z_i^2 \\
    &= \frac{1}{2}k_{d,1}^{(i-1)}Y_{i-1}^2 + k_{d,1}^{(i-1)}Y_{i-1}Z_i + \frac{1}{2}\big(k_{d,1}^{(i-1)} - k_{d,1}^{(i)}\big)Z_i^2.
\end{aligned}
\end{IEEEeqnarray} 
The average MAoI for the image modality is represented as
\begin{equation}
    {\bar\Delta}_{d,1}\! =\!\frac{\mathbb{E}[\frac{1}{2}k_{d,1}^{(i-1)}Y_{i\!-\!1}^2 \!+ \!k_{d,1}^{(i-1)}Y_{i-1}Z_{i} \!+\! \frac{1}{2}(k_{d,1}^{(i-1)} \!- \!k_{d,1}^{(i)})Z_{i}^2]}{\mathbb{E}(Y_{i-1})}.
\label{4}
\end{equation}

As we mentioned above, the devices generate status updates at a constant sampling interval $\tau_d$. Therefore, the inter-arrival time $Y_{i-1}$ is considered deterministic and satisfies
\begin{IEEEeqnarray}{c}
    \mathbb{E}(Y_{i-1}) = \tau_d.
\end{IEEEeqnarray}
From the PMF of $K_{d,1}$ in \eqref{pmf:k1}, we obtain
\begin{IEEEeqnarray}{c}
\mathbb{E}\!\left(K_{d,1}\right) = \mathbb{E}\!\left(k_{d,1}^{(i)}\right) = 1\! + \!\big(I_{d}^{\text{img}} \!+ \!R_{d}^{\text{img}}\big)\big(1\! -\! e^{-\lambda_1 \tau_d}\big).\IEEEeqnarraynumspace
\end{IEEEeqnarray}

Then, the term $\mathbb{E}\!\left[k_{d,1}^{(i-1)} Y_{i-1}^2\right]$ and $\mathbb{E}\!\left[k_{d,1}^{(i-1)} Y_{i-1} Z_i\right]$ can be calculated as
\begin{IEEEeqnarray}{c}
    \mathbb{E}\!\left[k_{d,1}^{(i-1)} Y_{i-1}^2\right] \!=\!  \left[1 + (I_{d}^{\text{img}} + R_{d}^{\text{img}})(1 - e^{-\lambda_1 \tau_d}) \right]\tau_d^2,\IEEEeqnarraynumspace
\end{IEEEeqnarray} 
\begin{IEEEeqnarray}{c}
    \mathbb{E}\!\left[k_{d,1}^{(i-1)} Y_{i-1} Z_i\right] \!= \! \left[1 \!+ \!(I_{d}^{\text{img}} \!+\! R_{d}^{\text{img}})(1\! - \!e^{-\lambda_1 \tau_d}) \right] \!\tau_d  Z_{i},\IEEEeqnarraynumspace
\end{IEEEeqnarray} 
where $Z_i$ denotes the system time of the $(i+1)$-th image data, computed from~\eqref{eq:system_time} as
\begin{equation}
    Z_i = T_{d,1}^{\text{sys}}(\mathbf{x}).
\end{equation}

Since $k_{d,1}^{(i-1)}$ and $k_{d,1}^{(i)}$ are i.i.d. random variables, the following formula can be calculated as
\begin{IEEEeqnarray}{c}
    \mathbb{E}\!\left[\tfrac{1}{2}\big(k_{d,1}^{(i-1)} - k_{d,1}^{(i)}\big) Z_i^2\right] = 0,
\end{IEEEeqnarray}

Therefore, the average MAoI for the image modality is given by
\begin{IEEEeqnarray}{c}
\begin{aligned}
{\bar\Delta}_{d,1} (\tau, \mathbf{x}) =& \left[1 + (I_{d}^{\text{img}} + R_{d}^{\text{img}})(1 - e^{-\lambda_1 \tau_d})\right] \\
&* \left[\frac{1}{2} \tau_d +  T_{d,1}^{\text{sys}}(\mathbf{x})\right].
\label{eq:total_aoi}
\end{aligned}
\end{IEEEeqnarray} 



\subsection{MAoI for the Audio Modality}
We design the MAoI for the audio modality based on two key characteristics: semantic changes and audio quality. 

\paragraph{Semantic Changes in Audio Content}
The timeliness of audio largely depends on its semantic content. 
The measure of semantic changes in audio, $V_{d}^{\text{sem}}$, is computed by averaging the semantic variations across consecutive audio frames
\begin{IEEEeqnarray}{c}
V_{d}^{\text{sem}} = \frac{1}{N_{\text{frame}}-1}\sum_{i=2}^{N_{\text{frame}
}}\frac{1}{M_{\text{aud}}}\sum_{k=1}^{M_{\text{aud}}}\left| f_{d,i}^{k}-f_{d,i-1}^{k}\right|,
\end{IEEEeqnarray} 
where $f_{d,i}^{k}$ denotes the $k$-th feature of the $i$-th frame, such as Mel-spectrum coefficients or Mel-frequency cepstral coefficients (MFCCs), $M_{\text{aud}}$ is the total number of feature dimensions, and $N_{\text{frame}}$ is the total number of frames.
\paragraph{Audio Quality}
The audio quality, denoted by $Q_d^{\text{aud}}$, is quantified by the product of the audio's sampling rate and bit depth
\begin{IEEEeqnarray}{c}
    Q_d^{\text{aud}} = H_d^{\text{aud}} \times U_d^{\text{aud}},
\end{IEEEeqnarray} 
where $H_d^{\text{aud}}$ represents the sampling rate, and $U_d^{\text{aud}}$ indicates the bit depth.

\paragraph{Event-Triggered Growth and Audio Average MAoI}
The MAoI growth rate $K_{d,2}$ for the audio modality is defined analogously, using a Poisson process with parameter $\lambda_2$ to represent events occurring within each sampling interval $\tau_d$. $K_{d,2}$ takes the higher value $1+V_{d}^{\text{sem}} + Q_d^{\text{aud}}$ when at least one audio event occurs within $\tau_d$, otherwise $K_{d,2} = 1$.
The MAoI for the audio modality is succinctly defined as
\begin{IEEEeqnarray}{c}
\Delta_{d,2}(t)=k_{d,2}^{(i)}\big(t-a_d(t)\big),\quad t\in(t_i,t_{i+1}],
\end{IEEEeqnarray}
{where }$k_{d,2}^{(i)}$ is the realization of $K_{d,2}$.

Since the derivation of the average MAoI for the audio modality follows a procedure analogous to that of the image modality, we directly present its resulting expression to avoid redundancy
\begin{IEEEeqnarray}{c}
\begin{aligned}
    \bar{\Delta}_{d,2}(\tau, \mathbf{x}) = & \left[1 + (V_{d}^{\text{sem}} + Q_d^{\text{aud}})(1 - e^{-\lambda_2 \tau_d})\right] \\
    &*\left[\frac{1}{2} \tau_d +  T_{d,2}^{\text{sys}}(\mathbf{x})\right],
\end{aligned}
\label{eq:total_aoi_audio}
\end{IEEEeqnarray}
where $T_{d,2}^{\text{sys}}(\mathbf{x})$ denotes the system time of the audio modality.

\subsection{MAoI for Other Signal Modalities}
Time-series measurements (e.g., environmental, physiological, or RF/radar) are continuously sampled and batched into frames for downstream processing. 
The MAoI for the other signal modalities captures two aspects: signal dynamics and Signal quality.

\paragraph{Signal Dynamics} 
Signal dynamics characterize how rapidly the content changes across frames and determine how quickly an observation becomes obsolete. 
Consider a segment with $L_d$ frames indexed by $i=1,\ldots, L_d$. In frame $i$ there are $N_{d,i}^{\text{sig}}$ detected points and each point provides $F_d^{\text{sig}}$ features. We form a frame-level feature vector by averaging the features over all detected points in that frame,
\begin{equation}
z_{d,i}^{(m)}=\frac{1}{N_{d,i}^{\text{sig}}}\sum_{n=1}^{N_{d,i}^{\text{sig}}} f_{d,i,n}^{(m)},\quad m=1,\ldots,F_d^{\text{sig}}.
\end{equation}
We then measure temporal variability by the mean squared difference between consecutive frame descriptors,
\begin{equation}
D_d^{\text{sig}}=\frac{1}{L_d-1}\sum_{i=2}^{L_d}\frac{1}{F_d^{\text{sig}}}\sum_{m=1}^{F_d^{\text{sig}}}\big(z_{d,i}^{(m)}-z_{d,i-1}^{(m)}\big)^2.
\end{equation}
This metric increases when the signal exhibits rapid transitions, while it remains small for slowly varying segments. 

\paragraph{Signal Quality}
Acquisition fidelity in time and amplitude defines the quality of the signal. We quantify signal quality as the product of the sampling rate and the quantization bit depth, which is
\begin{equation}
Q^{\text{sig}}_d = R^{\text{sig}}_d \cdot b^{\text{sig}}_d,
\end{equation}
where \(R^{\text{sig}}_d\) is the sampling rate and \(b^{\text{sig}}_d\) is the bit depth.

\paragraph{Event-Triggered Growth and Signal Average MAoI}
Consistent with the image and audio modalities, we model {signal events} within each sampling interval $\tau_d$ by a Poisson process with average generation rate $\lambda_3$. Let $K_{d,3}$ denote the MAoI growth rate for the other signal modalities of device $d$. It takes a higher value ($K_{d,3} = 1 + D^{\text{sig}}_d + Q^{\text{sig}}_d$) if at least one signal event occurs within $\tau_d$, and $1$ otherwise.
The MAoI for the other signal modalities is then
\begin{equation}
\Delta_{d,3}(t)=k_{d,3}^{(i)}\big(t-a_d(t)\big),\quad t\in\!\big(t_i,t_{i+1}\big],
\end{equation}
{where }$k_{d,3}^{(i)}$ denotes the growth rate for the $i$-th signal update and is the realization of $K_{d,3}$.
Following the same derivation used for images, the long-term average MAoI for the other signal modalities becomes
\begin{IEEEeqnarray}{c}
\begin{aligned}
    \bar{\Delta}_{d,3}(\tau, \mathbf{x}) = & \left[ 1 + (D^{\text{sig}}_d + Q^{\text{sig}}_d)(1 - e^{-\lambda_3 \tau_d}) \right] \\
     &*\left[ \frac{1}{2} \tau_d + T_{d,3}^{\text{sys}}(\mathbf{x})\right],
\end{aligned}
\label{eq:total_aoi_audio}
\end{IEEEeqnarray}
where $T_{d,3}^{\text{sys}}(\mathbf{x})$ is the system time for the other signal modality defined in~\eqref{eq:system_time}.

\subsection{Overall Average MAoI per Device}
The overall average MAoI for device $d$ integrates contributions from image, audio, and other signal modalities, which is expressed as in~(\ref{eq:total_average_taoi}), in which a weight $\Psi_{d,s}$ for device $d$ and modality $s$ is introduced to aggregate modality-specific attributes.
\begin{figure*}[!t]
\normalsize
\setcounter{equation}{56} 
\begin{IEEEeqnarray}{rCl}
        {\bar\Delta}_{d} (\tau_d, \mathbf{x}) &= &  {\bar\Delta}_{d,1} (\tau_d, \mathbf{x}) +  {\bar\Delta}_{d,2} (\tau_d, \mathbf{x}) +  {\bar\Delta}_{d,3} (\tau_d, \mathbf{x}) \nonumber  \\
        &=&  \left[1 + (I_{d}^{\text{img}} + R_{d}^{\text{img}})(1 - e^{-\lambda_1 \tau_d})\right]
* \left[\frac{1}{2} \tau_d +  T_{d,1}^{\text{sys}} (\mathbf{x})\right]\nonumber 
+  \left[1 + (Q_d^{\text{aud}} + V_{d}^{\text{sem}})(1 - e^{-\lambda_2 \tau_d})\right]
* \left[\frac{1}{2} \tau_d +  T_{d,2}^{\text{sys}} (\mathbf{x})\right] \\
&+&  \left[1 + (D^{\text{sig}}_d + Q^{\text{sig}}_d)(1 - e^{-\lambda_3 \tau_d})\right]
* \left[\frac{1}{2} \tau_d +  T_{d,3}^{\text{sys}} (\mathbf{x})\right] \nonumber \\
&=& \sum_{s=1}^{3}  \left[1 + \Psi_s(1 - e^{-\lambda_s \tau_d})\right]\left[\frac{1}{2}\tau_d + T_{d,s}^{\text{sys}}(\mathbf{x})\right],\Psi_{d,1}=I_{d}^{\text{img}}+R_{d}^{\text{img}},
\Psi_{d,2}=Q_d^{\text{aud}}+V_{d}^{\text{sem}},
\Psi_{d,3}=D^{\text{sig}}_d + Q^{\text{sig}}_d.\label{eq:total_average_taoi}
\end{IEEEeqnarray}
\hrulefill
\end{figure*}
\subsection{Problem Formulation}
The primary objective is to minimize the system average MAoI across all IoT devices, while ensuring that the average energy consumption of each device remains within a specified budget.
Specifically, the average energy consumption per unit time for IoT device $d$ is
\begin{IEEEeqnarray}{c}
\begin{aligned}
\bar{E}_d (\tau_d, \mathbf{x}) &= \frac{E_d(\mathbf{x})}{\mathbb{E}(Y_i)} = \frac{{E_d(\mathbf{x})}}{\tau_d}.
\label{eq:energy_total}
\end{aligned}
\end{IEEEeqnarray}
Accordingly, the optimization problem is formulated as
\begin{subequations} \label{eq:problem_constraint}
\begin{align}
\mathcal{(P)} : \: & \min_{(\tau, \mathbf{x})}  \sum_{d=1}^{D}{\bar{\Delta}_d (\tau_d, \mathbf{x})}  \label{eq:objective_new} \\
\text{s.t.} \: 
& x_d \in \{0,1\}, \quad \forall d \in \mathcal{D}\label{eq:binary_constraint_new} \\
& \tau_d \ge \tau_{\min}, \quad \forall d \in \mathcal{D} \label{eq:tau_constraint_new} \\
& \sum_{d=1}^{D} x_d (S_{d}^{\text{img}} + S_{d}^{\text{aud}} + S_{d}^{\text{sig}})\le S_{\text{th}}, \label{eq:capacity_constraint_new} \\
& \bar{E}_d (\tau_d, \mathbf{x}) \leq E_{d}^{\max}, \quad \forall d \in \mathcal{D} \label{eq:energy_constraint}
\end{align}
\end{subequations}
where $E_{d}^{\max}$ is the average energy budget for device $d$. Constraint~\eqref{eq:binary_constraint_new} ensures binary offloading decisions, constraint~\eqref{eq:tau_constraint_new} imposes the minimum feasible sampling interval, constraint~\eqref{eq:capacity_constraint_new} enforces the data capacity limit of the MEC server, and constraint~\eqref{eq:energy_constraint} guarantees that the average energy consumption of each device does not exceed its energy budget.

Problem $\mathcal{(P)}$ is challenging because the energy constraint couples the sampling intervals $\tau$ and the offloading decisions $\mathbf{x}$, which complicates the joint optimization. 
Accordingly, we relax~\eqref{eq:energy_constraint} using Lagrangian multipliers, which incorporates energy consumption into the objective, and obtain the following relaxed formulation 
\begin{align}
\mathcal{(P')}: \; &\min_{(\tau, \mathbf{x})}  \sum_{d=1}^{D}\left[{\bar{\Delta}_d (\tau_d, \mathbf{x}) + \mu_d \left( \bar{E}_d (\tau_d, \mathbf{x}) - E_{d}^{\max} \right)}\right],\\
 \text{s. t.} \; & \eqref{eq:binary_constraint_new},~\eqref{eq:tau_constraint_new},~\eqref{eq:capacity_constraint_new} \nonumber.
\label{eq:lagrange_obj}
\end{align}
where $\mu_d \geq 0$ is the Lagrange multiplier associated with the energy constraint of device $d$.

\section{The proposed MAoI Minimization Algorithm}\label{sec:algorithm}
Under per-device energy budgets, the MAoI minimization is approached by relaxing the average-energy constraints via Lagrangian dualization, yielding a penalized objective for a relaxed problem. 
Building on this, the mixed continuous–discrete problem is decomposed into a sampling-interval subproblem and an offloading subproblem.
We then propose the JSO algorithm that alternates an optimal sampling-interval subalgorithm with an interference-aware best-response offloading subalgorithm, while updating the Lagrange multipliers to enforce the energy budgets. These three updates repeat until convergence, yielding a practical joint solution for sampling and offloading.

\subsection{Problem Decomposition}
Given the coupling of continuous and discrete variables, we decompose the original joint optimization problem into two sequential subproblems.

\subsubsection{Sampling Interval Optimization}
For given computation offloading decisions $\mathbf{x}$, the sampling interval optimization subproblem for all devices is formulated as
\begin{align}
(\mathcal{P}_1'): \; &\min_{\tau}  \sum_{d=1}^{D} \left[ \bar{\Delta}_d (\tau_d, \mathbf{x}) + \mu_d \left( \bar{E}_d (\tau_d, \mathbf{x}) - E_{d}^{\max} \right) \right], \\ \text{s.t.} \; & \eqref{eq:binary_constraint_new}\nonumber.
\label{eq:subproblem1_relaxed}
\end{align}

\subsubsection{Computation Offloading Optimization}
With the sampling intervals determined by solving $(\mathcal{P}_1')$, the computation offloading optimization subproblem is formulated as
\begin{align}
(\mathcal{P}_2'): \; &\min_{\mathbf{x}}  \sum_{d=1}^{D} \left[ \bar{\Delta}_d (\tau_d, \mathbf{x}) + \mu_d \left( \bar{E}_d (\tau_d, \mathbf{x}) - E_{d}^{\max} \right) \right], \\
\text{s.t.} \; &~\eqref{eq:tau_constraint_new},~\eqref{eq:capacity_constraint_new}.\nonumber
\label{eq:subproblem2_relaxed}
\end{align}

\subsection{Optimal Sampling-Interval Subalgorithm}
Given fixed offloading decisions and Lagrange multipliers, the sampling interval optimization subproblem $(\mathcal{P}_1')$ can be decomposed into individual device-level subproblems. Specifically, for device $d$, the objective function is
\begin{align}
\bar{C}_d(\tau_d) =
& \sum_{s=1}^{3} \left[1 + \Psi_s (1 - e^{-\lambda_s \tau_d})\right]\left[\frac{1}{2}\tau_d + T_{d,s}^{\text{sys}}(\mathbf{x})\right] \nonumber \\
&~+ \mu_d\left(\frac{E_d(\mathbf{x})}{\tau_d} - E_d^{\max}\right).
\end{align}


To analyze the optimality conditions, we compute the first and second derivatives of the objective function with respect to $\tau_d$
\begin{align}
\frac{\partial \bar{C}_d(\tau_d)}{\partial \tau_d} =\,&
\sum_{s=1}^{3} \Psi_s \lambda_s e^{-\lambda_s \tau_d} T_{d,s}^{\text{sys}}(\mathbf{x}) 
+ \frac{\tau_d}{2} \sum_{s=1}^{3} \Psi_s \lambda_s e^{-\lambda_s \tau_d} \nonumber \\
&+ \frac{1}{2}\sum_{s=1}^{3}\left[1+\Psi_s(1-e^{-\lambda_s\tau_d})\right]
- \mu_d \frac{E_d(\mathbf{x})}{\tau_d^2},
\end{align}
\begin{equation}
\frac{\partial^2 \bar{C}_d(\tau_d)}{\partial \tau_d^2}\! = \!\mu_d \frac{2E_d(\mathbf{x})}{\tau_d^3}
+ \sum_{s=1}^{3} \Psi_s \lambda_s e^{-\lambda_s \tau_d}\!\left[1\! -\! \frac{\lambda_s \tau_d}{2}\! - \! \lambda_s T_{d,s}^{\text{sys}}(\mathbf{x})\right].
\end{equation}
The convexity of the objective function is determined by the sign of its second derivative. Given that $({2E_d}/{\tau_d^3})>0$ and $\sum_{s=1}^{3}\Psi_s\lambda_se^{-\lambda_s\tau_d}>0$, strict convexity occurs when
\begin{equation}
1-\frac{\lambda_s\tau_d}{2}-\lambda_sT_{d,s}^{\text{sys}}(\mathbf{x}) \ge 0, \quad \forall s\in\{1,2,3\},
\end{equation}
which implies
\begin{equation}
\tau_d\le\tau_d^{\text{th}}, \;\text{where}\;\tau_d^{\text{th}}=\min_{s\in{1,2,3}}\left\{\frac{2(1-\lambda_sT_{d,s}^{\text{sys}}(\mathbf{x}))}{\lambda_s}\right\}.
\label{eq:threshold_interval}
\end{equation}

Let $\tau_d^{\text{upper}} = \max\{\tau_{\text{min}}, \tau_d^{\text{th}}\}$. When $\tau_d> \tau_d^{\text{upper}}$, the convexity of the original objective function $\bar{C}_d(\tau_d)$ cannot be guaranteed. Thus, we introduce a strictly convex upper bound function $G_d(\tau_d)$, which is
\begin{align}
G_d(\tau_d)=&\sum_{s=1}^{3}  \left[1 + \Psi_s(1 - e^{-\lambda_s \tau_d^{\text{upper}}})\right] \nonumber \\
    &*\left[\frac{1}{2}\tau_d + T_{d,s}^{\text{sys}}(\mathbf{x})\right] + \mu_d\left(\frac{E_d}(\mathbf{x}){\tau_d} - E_d^{\max}\right).
    \label{eq:upper_bound_function}
\end{align}

The inequality $G_d(\tau_d)\ge \bar{C}_d(\tau_d)$ holds because the exponential terms in $G_d(\tau_d)$ are evaluated at the smaller fixed interval $\tau_d^{\text{up}}$, resulting in larger exponential values due to the monotonicity of the exponential function.

The strict convexity of $G_d(\tau_d)$ is confirmed by its second derivative
\begin{equation}
G_d''(\tau_d)= \frac{2\mu_d E_d(\mathbf{x})}{\tau^3} >0 .
\end{equation}
Subsequently, the closed-form solution minimizing $G_d(\tau_d)$ can be obtained by setting its first derivative to zero, i.e., $G_d'(\tau_d)=0$. The optimal sampling interval for $G_d(\tau_d)$ is thus given by
\begin{equation}
\tau_d^{\text{sub}}=\sqrt{\frac{2\mu_d E_d(\mathbf{x})}{\sum_{s=1}^{3}\left[1+\Psi_s\left(1-e^{-\lambda_s\tau_d^{\text{upper}}}\right)\right]}}.
\label{eq:optimal_convex_solution}
\end{equation}
Then, we obtain a feasible approximation
\begin{equation}
\tau_d^{\text{approx}}=\max\left\{\tau_d^{\text{th}},\tau_{\text{min}},\tau_d^{\text{sub}}\right\}.
\label{eq:final_optimal_interval}
\end{equation}
If $\tau_{\text{min}} < \tau_d^{\text{th}}$, we additionally compute the optimal solution within the convex region $\tau_d \in [\tau_{\text{min}}, \tau_d^{\text{th}}]$. In this region, the objective function $\bar{C}_d(\tau_d)$ is strictly convex. Hence, a unique global optimum can be determined by numerically solving the first-order optimality condition
\begin{equation}
\frac{\partial\bar{C}_{d}(\tau_d)}{\partial\tau_d}=0.
\end{equation}
Given the complexity of directly obtaining an analytical solution, we apply Newton's iterative method as follows
\begin{equation}
\tau_d^{(n+1)} = \tau_d^{(n)} - \frac{\frac{\partial \bar{C}_d}{\partial \tau_d}(\tau_d^{(n)})}{\frac{\partial^2 \bar{C}_d}{\partial \tau_d^2}(\tau_d^{(n)})},
\end{equation}
where at each iteration, the updated value $\tau_d^{(n+1)}$ is projected onto the feasible interval $[\tau_{\text{min}}, \tau_d^{\text{th}}]$, i.e.,
\begin{equation}
\tau_d^{(n+1)} \leftarrow \min\left\{\max\left\{\tau_d^{(n+1)},\ \tau_{\text{min}}\right\},\ \tau_d^{\text{th}}\right\}.
\label{eq:newton_method}
\end{equation}
This projection ensures that all intermediate and final solutions produced by Newton’s method remain within the feasible region. The iteration proceeds until convergence to a stable solution $\tau_d^{\text{newton}}$ within $[\tau_{\text{min}}, \tau_d^{\text{th}}]$.


Finally, the optimal sampling interval $\tau_d^*$ is selected as
\begin{equation}
\tau_d^*=\arg\min_{{\tau_d^{\text{newton}},\tau_d^{\text{approx}}}}\bar{C}_d(\tau_d),
\label{eq:chose_tau_star}
\end{equation}
where $\tau_d^{\text{newton}}$ is considered only if the convex region exists, i.e., $\tau_{\text{min}} < \tau_d^{\text{th}}$. The detailed procedure is summarized in Algorithm~\ref{alg:sampling_interval}.

\begin{algorithm}[t]
\caption{Optimal Sampling-Interval Subalgorithm}
\label{alg:sampling_interval}
\begin{algorithmic}[1]
\REQUIRE System parameters $\mathbf{x}$, $\mu_d$ $\Psi_s$, $\lambda_s$, $T_{d,s}^{\text{sys}}(\mathbf{x})$, $E_d$, $E_d^{\max}$, and $\tau_{\text{min}}$ for each device $d \in \mathcal{D}$
\ENSURE Optimal sampling intervals $\{\tau_d^*\}_{d\in\mathcal{D}}$
\FOR{each device $d \in \mathcal{D}$}
    \STATE Compute threshold $\tau_d^{\text{th}}$ using \eqref{eq:threshold_interval}
    \STATE Set $\tau_d^{\text{upper}} = \max\{\tau_{\text{min}}, \tau_d^{\text{th}}\}$
    \STATE Compute suboptimal solution $\tau_d^{\text{sub}}$ using \eqref{eq:optimal_convex_solution}
    \STATE Compute feasible approximation $\tau_d^{\text{approx}}$ using~\eqref{eq:final_optimal_interval}
    \IF{$\tau_{\text{min}} < \tau_d^{\text{th}}$}
        \STATE Initialize $\tau_d^{(0)} \in [\tau_{\text{min}}, \tau_d^{\text{th}}]$
        \REPEAT
            \STATE Update $\tau_d^{(n+1)}$ via Newton's method~\eqref{eq:newton_method}
        \UNTIL{convergence to $\tau_d^{\text{newton}}$}
        \STATE Evaluate $\bar{C}_d(\tau_d^{\text{newton}})$ and $\bar{C}_d(\tau_d^{\text{approx}})$
        \STATE Select $\tau_d^*$ using~\eqref{eq:chose_tau_star}
    \ELSE
        \STATE Set $\tau_d^* = \tau_d^{\text{approx}}$
    \ENDIF
\ENDFOR
\RETURN $\{\tau_d^*\}_{d\in\mathcal{D}}$
\end{algorithmic}
\end{algorithm}

\subsection{Interference-Aware Best-Response Offloading Subalgorithm}
In this subsection, the computation offloading optimization subproblem $(\mathcal{P}_2')$ is addressed, aiming to minimize the overall system cost given fixed sampling intervals and Lagrange multipliers. Formally, the simplified optimization problem is expressed as
\begin{align}
(\mathcal{P}_2'): &\min{\mathbf{x}} \sum_{d=1}^{D}{\bar{C}_d (\mathbf{x})},  \\
\text{s.t.}\; &x_d \in \{0,1\}, \quad \forall d \in \mathcal{D}, \nonumber\ \\
&\sum_{d=1}^{D} x_d (S_{d}^{\text{img}} + S_{d}^{\text{aud}} + S_{d}^{\text{txt}})\le S_{\text{th}},\nonumber
\end{align}
where the system cost $\bar{C}_d(\mathbf{x})$ is simplified by defining an auxiliary parameter $\phi_s$, which is
\begin{equation}
\phi_s = \left[1 + \Psi_s\left(1 - e^{-\lambda_s \tau_d}\right)\right].
\end{equation}
Consequently, the system cost function becomes
\begin{align}
{\bar{C}_d (\mathbf{x})} \!=\! & \sum_{s=1}^{3} \phi_s\!\left(\frac{1}{2}\tau_d \!+\! T_{d,s}^{\text{sys}}(\mathbf{x})\right) \!+\! \mu_d\!\left(\frac{E_d}{\tau_d} \!- \!E_d^{\max}\right).
\end{align}

The above optimization is equivalent to a maximum cardinality bin packing problem, which is NP-hard~\cite{tresca2022automating}. Thus, centralized solutions become computationally infeasible for large-scale IoT deployments. To address this computational challenge, we adopt a decentralized game-theoretic framework, leveraging the natural capability of game theory to handle distributed and independent decision-making scenarios.

We formally define a strategic game $\mathcal{G} = \{\mathcal{D},\mathbf{x}, \bar{C}_{d}\}$, where each IoT device $d\in \mathcal{D}$ constitutes a player, the offloading decision represents each player's strategy, and the cost function $\bar{C}_d$ denotes each player's individual cost. In this setting, each device independently solves its optimization problem
\begin{equation}
\min_{x_d\in\{0,1\}}{\bar{C}_d (x_d,x_{-d})},\quad \forall d \in \mathcal{D},
\end{equation}
where $x_{-d}$ represents the strategies of other signal devices.
To effectively solve the decentralized optimization problem outlined above, we seek a stable solution characterized by the Nash equilibrium, defined as follows.
\begin{definition}
A strategy set $\mathbf{x}^* = \{x_1^*,\dots,x_D^*\}$ constitutes a Nash equilibrium for the computation offloading problem if no device can unilaterally reduce its cost by altering its offloading strategy, formally expressed as
\begin{equation}
\bar{C}_d (x_d^*,x_{-d}^*) \le \bar{C}_d (x_d,x_{-d}^*), \quad \forall d\in \mathcal{D}.
\end{equation}
\end{definition}

To efficiently reach the Nash equilibrium, devices iteratively select their best-response strategies based on the current strategies of other signal devices.
\begin{definition}
Given the strategies of other signal devices $x_{-d}^*$, the best response strategy $x_d^*$ for device $d$ is defined as
\begin{equation}
    x_d^* = \arg\min_{x_d \in \{0,1\}} \bar{C}_d (x_d,x_{-d}^*).
\end{equation}
\end{definition}

The best-response condition for offloading decisions is explicitly characterized by the following lemma
\begin{lemma}
Device $d$ prefers offloading status updates to the MEC server if the interference from other signal devices satisfies
\begin{equation}
\sum_{j\neq d}^{D} x_{j}P_{j}g_{j}^{\text{edge}} \le \Lambda_d,
\end{equation}
where
\begin{equation}
\Lambda_d=\frac{P_{d}g_{d}^{\text{edge}}}{2^{\frac{(\sum_{s=1}^{3}\phi_s\tau_d+{\mu_e P_d})S_d}{B\sum_{s=1}^{3}\phi_s\tau_d(W_{d,s} + T_{d,s}^{\text{local,comp}}-T_{d,s}^{\text{edge,comp}})+ \mu_e E_d^{\text{comp}}}}-1}-\omega_0.
\end{equation}
\label{Lm:offload_choice}
\end{lemma}
\begin{proof}
See Appendix~A in the supplemental material.
\end{proof}
Thus, the best response decision for each device $d$ is 
\begin{equation}
x_d =
\begin{cases}
1, &\text{if } \sum_{j\neq d}^{D} x_{j}P_{j}g_{j}^{\text{edge}} \le \Lambda_d,\\
0, &\text{otherwise.}
\end{cases}
\end{equation}

Then, we formally establish the existence of a Nash equilibrium in the proposed computation offloading game through the following corollary derived from potential game theory
\begin{corollary}\label{cor:offloading}
The formulated computation offloading game is a potential game, ensuring the existence of at least one Nash equilibrium and satisfying the finite improvement property.
\end{corollary}
\begin{proof}
See Appendix~B in the supplemental material.
\end{proof}
Corollary~\ref{cor:offloading} ensures that our computation offloading optimization problem converges to a Nash equilibrium within finite iterative steps.

In practical implementation, an interference-aware best-response offloading subalgorithm is employed, which is summarized in Algorithm~\ref{alg:offloading}. At iteration $k$, with all other devices strategies fixed and $(\boldsymbol{\tau},\boldsymbol{\mu})$ given, each device $d$ computes its best response $x_d^{\text{BR}}$ to $\mathbf{x}^{(k)}_{-d}$ according to Lemma~1. If $x_d^{\text{BR}} \neq x_d^{(k)}$, let $\mathbf{x}^{(d)}$ denote the profile obtained from $\mathbf{x}^{(k)}$ by replacing $x_d^{(k)}$ with $x_d^{\text{BR}}$. The corresponding reduction in the system cost is
\begin{equation}
\Upsilon _d \triangleq \bar{C}\big(\mathbf{x}^{(k)}\big) - \bar{C}\big(\mathbf{x}^{(d)}\big).
\label{eq:br_reduction}
\end{equation}
Among all indices with $\Upsilon _d>0$, the update with the largest $\Upsilon _d$ is committed, and the remaining decisions are left unchanged. Each committed update is a best response and strictly decreases the potential. The potential is monotonically nonincreasing, and the finite-improvement property of potential games ensures convergence to a Nash equilibrium in a finite number of iterations.

\begin{algorithm}[t]
\caption{Interference-Aware Best-Response Offloading Subalgorithm}
\label{alg:offloading}
\begin{algorithmic}[1]
\REQUIRE System parameters $\tau_d$, $\mu_d$, $\Psi_s$, $\lambda_s$, $T^{\text{sys}}_{d,s}(\mathbf{x})$, $E_d(\mathbf{x},\boldsymbol{\tau})$, $E_d^{\max}$, and $\tau_{\min}$ for each $d \in \mathcal{D}$
\STATE \textbf{Initialize:} Choose initial $\mathbf{x}^{(0)}$ and set $k=0$
\REPEAT
    \STATE Compute $\widehat{C}\big(\mathbf{x}^{(k)}\big)$
    \STATE Set $d^\star \leftarrow \varnothing$ and $\Upsilon ^\star \leftarrow 0$
    \FOR{each device $d \in \mathcal{D}$}
        \STATE Compute $x_d^{\text{BR}}$ from Lemma~\ref{Lm:offload_choice} given $\mathbf{x}^{(k)}_{-d}$
        \IF{$x_d^{\text{BR}} \neq x_d^{(k)}$}
            \STATE Let $\mathbf{x}^{(d)}$ be $\mathbf{x}^{(k)}$ with $x_d \leftarrow x_d^{\text{BR}}$
            \STATE Evaluate $\Upsilon _d$ using \eqref{eq:br_reduction}
            \IF{$\Upsilon _d > \Upsilon ^\star$}
                \STATE $d^\star \leftarrow d$ and $\Upsilon ^\star \leftarrow r_d$
            \ENDIF
        \ENDIF
    \ENDFOR
    \IF{$\Upsilon ^\star > 0$}
        \STATE $\mathbf{x}^{(k+1)} \leftarrow \mathbf{x}^{(d^\star)}$ and $k \leftarrow k+1$
    \ENDIF
\UNTIL{No improving best response exists; $\Upsilon^\star \le 0$}
\RETURN $\mathbf{x}^\star \leftarrow \mathbf{x}^{(k)}$
\end{algorithmic}
\end{algorithm}


\subsection{Lagrange Multiplier Update}
After solving for the sampling intervals and offloading strategies in each iteration, the Lagrange multipliers $\mu_d$ are updated using the subgradient method to enforce the average energy constraint. For each device $d$, the update rule is
\begin{equation}
\mu_d^{(k+1)} = \left[ \mu_d^{(k)} + \eta^{(k)} \left( \bar{E}_d \big(\tau_d^{(k)},\, \mathbf{x}^{(k)}\big) - E_d^{\max} \right) \right]^+,
\label{eq:lambda_update}
\end{equation}
where $k$ is the iteration index, $\eta^{(k)}$ is the step size, and $[\cdot]^+$ denotes the projection onto the non-negative orthant.

This iterative update is performed jointly with the alternating optimization of sampling intervals and offloading strategies, and is repeated until both the objective function and the energy constraints converge within prescribed tolerances.
\begin{algorithm}[t]
\caption{Joint Sampling Offloading Optimization Algorithm}
\label{alg:jaso}
\begin{algorithmic}[1]
\REQUIRE System parameters for all devices; initial sampling intervals $\{\tau_d^{(0)}\}$; initial offloading vector $\mathbf{x}^{(0)}$; initial Lagrange multipliers $\{\mu_d^{(0)}\}$; convergence threshold $\epsilon$
\STATE Set iteration index $k \gets 0$
\REPEAT
    \STATE Optimal sampling-interval subalgorithm (Algorithm~\ref{alg:sampling_interval}):\\
    \hspace{1.5em} Given $\mathbf{x}^{(k)}$ and $\{\mu_d^{(k)}\}$, update $\{\tau_d^{(k+1)}\}$
    \STATE Interference-aware best-response offloading subalgorithm (Algorithm~\ref{alg:offloading}):\\
    \hspace{1.5em} Given $\{\tau_d^{(k+1)}\}$ and $\{\mu_d^{(k)}\}$, update $\mathbf{x}^{(k+1)}$
    \STATE Lagrange multiplier update (Eq.~\eqref{eq:lambda_update}):\\
    \hspace{1.5em} For each $d$, update $\lambda_d^{(k+1)}$
    \STATE Compute overall system cost $\hat{C}^{(k+1)} = \sum_{d=1}^{D} C_d^{(k+1)}$
    \STATE $k \gets k + 1$
\UNTIL{$|\hat{C}^{(k)} - \hat{C}^{(k-1)}| < \epsilon$}
\RETURN Optimal sampling intervals $\{\tau_d^*\}$, offloading strategy $\mathbf{x}^*$, and Lagrange multipliers $\{\mu_d^*\}$
\end{algorithmic}
\end{algorithm}

\subsection{Complexity Analysis}
The details of the proposed JSO algorithm are summarized in Algorithm~\ref{alg:jaso}. The iterative process is guaranteed to converge, as the system cost is lower-bounded and decreases monotonically with each iteration.

The computational complexity of the proposed JSO algorithm is mainly determined by its three iterative modules. At each iteration, the optimal sampling-interval subalgorithm is performed independently for each device, and the per-device solution via Newton's method requires at most $N$ updates. Hence, the total complexity of Algorithm~\ref{alg:sampling_interval} is
\begin{IEEEeqnarray}{c}
    \mathcal{T}_1 = \mathcal{O}(DN).
\end{IEEEeqnarray}
For the interference-aware best-response offloading subalgorithm, distributed best-response dynamics are adopted. Each device sequentially evaluates its offloading decision by calculating aggregate interference from all other signal devices, and $I$ denotes the number of best-response rounds. Therefore, the complexity of Algorithm~\ref{alg:offloading} is
\begin{IEEEeqnarray}{c}
    \mathcal{T}_2 = \mathcal{O}(ID^2).
\end{IEEEeqnarray}
The Lagrange multiplier update involves only simple arithmetic operations and incurs negligible overhead. As a result, the overall per-iteration complexity of Algorithm~\ref{alg:jaso} is
\begin{IEEEeqnarray}{c}
    \mathcal{T}_3^{\text{iter}} = \mathcal{O}(DN + ID^2),
\end{IEEEeqnarray}
Therefore, the total complexity of JSO algorithm over $K$ outer iterations is
\begin{IEEEeqnarray}{c}
    \mathcal{T}_3 = \mathcal{O}\big(K[DN + ID^2]\big).
\end{IEEEeqnarray}

\begin{table}[t]
\centering
\renewcommand{\arraystretch}{1.15}
\caption{Simulation Parameters.}
\label{tab:total_params}
\begin{tabular}{lll}
\toprule
\textbf{Parameter} & \textbf{Symbol} & \textbf{Value/Setting} \\
\midrule
Transmit power & $P_d$ & $0.1~\mathrm{W}$ \\
Channel bandwidth & $B$ & $1~\mathrm{MHz}$ \\
Noise power & $\omega_0$ & $-100~\mathrm{dBm}$ \\
Local CPU frequency & $f_{\text{local}}$ & $1~\mathrm{GHz}$ \\
Edge CPU frequency & $f_{\mathrm{edge}}$ & $10~\mathrm{GHz}$ \\
Computation energy coefficient & $\gamma$ & $10^{-9}~\mathrm{J/FLOP}$ \\
Max device energy & $E_{d}^{\max}$ & $1~\mathrm{J}$ \\
Min. sampling interval & $\tau_{\min}$ & $2~\mathrm{s}$ \\
Image arrival rate  & $\lambda_{d}^{\text{img}}$ & $0.8~\mathrm{s}^{-1}$ \\
Audio arrival rate & $\lambda_{d}^{\text{aud}}$ & $0.8~\mathrm{s}^{-1}$ \\
Signal arrival rate  & $\lambda_{d}^{\text{sig}}$ & $0.8~\mathrm{s}^{-1}$ \\
Lagrange step size & $\eta$ & $0.01$ \\
Convergence threshold & $\epsilon$ & $10^{-3}$ \\
Image data size & $H_d\times W_d$ & $224\times 224~\mathrm{pixels}$ \\
Audio duration & $t_{d}^{\text{aud}}$ & $2~\mathrm{s}$ \\
Audio sampling rate & $R_{d}^{\text{aud}}$ & $16~\mathrm{kHz}$ \\
Audio bit depth & $b_{d}^{\text{aud}}$ & $16~\mathrm{bits}$ \\
Audio channels & $n_{d}^{\text{aud}}$ & $1$ (mono) \\
Signal duration & $t_{d}^{\text{sig}}$ & $3~\mathrm{s}$ \\
Signal sampling rate & $R_{d}^{\text{sig}}$ & $80~\mathrm{Hz}$ \\
Signal bit depth & $b_{d}^{\text{sig}}$ & $16~\mathrm{bits}$ \\
Camera per-capture overhead & $\kappa^{\text{cam}}_{d}$ & $5\,\mathrm{mJ}$ \\
Per-pixel energy & $e^{\text{pix}}_{d}$ & $15\,\mathrm{pJ/pixel}$ \\
Image channels & $c^{\text{img}}$ & $3$ \\
Audio baseline power & $P^{\text{aud}}_{d}$ & $8\,\mathrm{mW}$ \\
ADC/I/O scaling & $\alpha_{d}$ & $1.0\times 10^{-8}$ \\
Radar active power & $P^{\text{sig}}_{d}$ & $50\,\mathrm{mW}$ \\
\bottomrule
\end{tabular}
\label{tab:sim_params}
\end{table}

\section{Experimental Evaluation}\label{sec:experimental}
\subsection{Simulation Setup}
The MEC system consists of a single MEC server co-located with BS covering a \num{40} m $\times$ \num{40} m area, within which $D$ IoT devices are randomly distributed. The uplink channel gain between each device and a BS is modeled as $g_{d}^{\mathrm{edge}} = h_{d}^{-\delta}$, where $h_{d}$ denotes the device-BS distance and the path-loss exponent $\delta$ is set to \num{2}.

{Image data} is processed using the ResNet-18 architecture~\cite{he2016resnet}, following the standard implementation in the {torchvision} library. The network consists of \num{17} convolutional layers and \num{1} fully connected layer, grouped into four residual blocks. The model contains approximately \num{11.7}~million parameters and is pretrained on the ImageNet dataset. When processing input images of size $224 \times 224$ pixels with \num{3}~channels, the average computational complexity is approximately {4 GFLOPs} per forward pass, based on standard FLOPs profiling tools.

{Audio data} is processed using the DeepSpeech2 model~\cite{amodei2016deep}, based on the open-source PyTorch implementation~\cite{deepspeechpytorch}. The model consists of \num{2} convolutional layers (with kernel sizes $[11 \times 41]$ and $[11 \times 21]$, stride $[2 \times 2]$), followed by 5 bidirectional GRU layers (with 800 hidden units each), and a final fully connected layer for classification. Input audio is transformed into spectrograms using a Hamming window of size 20~ms and stride 10~ms, resulting in approximately 100 frames per second of audio. Empirical complexity profiling estimates the model's computational cost to be around 5~GFLOPs per second of input audio.

{Other signal data} is processed using the TFT~\cite{lim2021temporal}, based on the open-source PyTorch implementation~\cite{pytorchforecasting}. The TFT model comprises four encoder layers with a hidden dimension of 160, 8 attention heads, and a feed-forward dimension of 640. The model integrates multi-head attention and position-wise feed-forward mechanisms. When processing an input sequence of 200 frames, the estimated computational complexity is approximately {0.45 GFLOPs} per sequence. 

The remaining simulation parameters are summarized in Table~\ref{tab:sim_params}, in which, we set identical arrival rates for image, audio, and other signal modalities, $\lambda_{d}^{\text{img}}=\lambda_{d}^{\text{aud}}=\lambda_{d}^{\text{sig}}=0.8\,\mathrm{s}^{-1}$, to keep the arrival processes comparable across modalities and to focus the evaluation on MAoI and the proposed JSO algorithm. 

\begin{figure}
    \centering
    \includegraphics{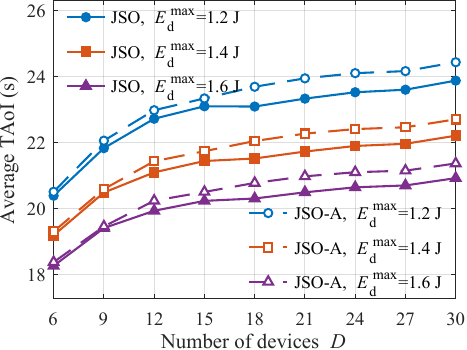}
    \caption{System average MAoI vs. the number of devices for JSO and JSO-A under different energy budgets}
    \label{fig:taoi_vs_D}
\end{figure}

\begin{figure}
    \centering
    \includegraphics{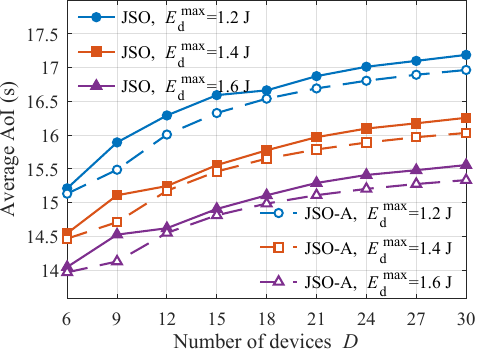}
    \caption{System average AoI vs. the number of devices for JSO and JSO-A under different energy budgets}
    \label{fig:aoi_vs_D}
\end{figure}

\begin{figure}
\centering\includegraphics{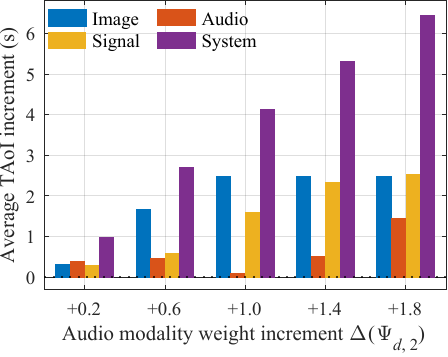}
    \caption{Per-modality and system average MAoI increments vs. the audio modality weight increment ($\Delta\Psi_{d,2}$) under fixed $D$ and $E_d^{\max}$. }
    \label{fig:modality_increment_TAoI}
\end{figure}
\begin{figure}
    \centering
    \includegraphics{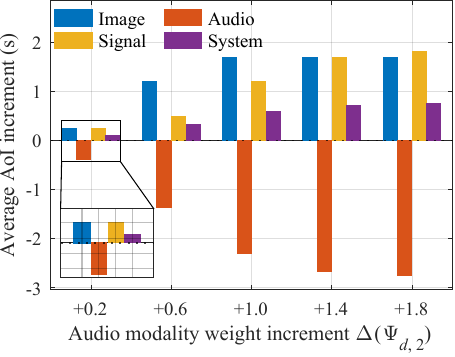}
    \caption{Per-modality and system average AoI increments vs. the audio modality weight increment ($\Delta\Psi_{d,2}$) under fixed $D$ and $E_d^{\max}$.}
    \label{fig:4_modality_increment_AoI}
\end{figure}

\subsection{Effectiveness of the MAoI}
To quantify the benefit of the MAoI, we run a controlled comparison where the optimization pipeline is fixed and only the objective changes. 
We construct a modality-agnostic baseline, {JSO-A}, by replacing the MAoI objective with the average AoI under the same pipeline and constraints. 
When all components are matched and only the objective differs, the observed performance difference isolates the effect of encoding modality weights in MAoI relative to standard AoI.


Figs.~\ref{fig:taoi_vs_D} and \ref{fig:aoi_vs_D} plot the system average MAoI and AoI versus the number of devices $D$, with curves for {JSO} and {JSO-A} under several energy budgets $E_d^{\max}$. 
Across all $D$ and $E_d^{\max}$, {JSO} attains a lower system-average MAoI than {JSO-A}, while {JSO-A} attains a lower system-average AoI under the same settings. The difference stems from the optimization targets. {JSO} optimizes MAoI, which embeds modality weights $\Psi_{d,s}$ into the age metric and therefore prioritizes faster or more important modalities; this can increase the system-average AoI of slower modalities even as AoI decreases. {JSO-A} optimizes the unweighted AoI and treats modalities uniformly, which lowers AoI but does not exploit modality-aware freshness, so its MAoI remains higher.
The results show that, unlike modality-agnostic AoI, the proposed MAoI is a modality-aware metric that captures per-modality dynamics and timeliness, making it well-suited for measuring information freshness in multimodal systems.

\subsection{Impact of the Modality Weight}
We characterize how increasing the audio weight affects per-modality and system MAoI and AoI with all other signal weights and system settings fixed. This experiment highlights MAoI's role in capturing modality importance and freshness by examining how per-modality and system MAoI respond to an additive increase in the audio weight $\Delta\Psi_{d,2}$ under fixed settings.
In our formulation, modality weights are independent coefficients and are not normalized to sum to one; therefore, we adjust the audio weight additively by $\Delta\Psi_{d,2}$ without re-normalization to isolate the marginal effect of the audio weight. Any changes for image and signal arise from queueing interactions rather than from explicit reductions of their weights.

Figs.~\ref{fig:modality_increment_TAoI} and~\ref{fig:4_modality_increment_AoI} report the per-modality and system-level increments in MAoI and AoI as the audio modality weight is increased from a randomly initialized baseline. Here, the baseline denotes the initial randomly initialized weight setting used to start the experiment, and the increments are measured relative to this setting.

Under fixed resources, a larger audio weight raises the service priority of audio updates, shortens their local waiting time, and thereby reduces their AoI. The audio MAoI can vary non-monotonically because two effects act in opposite directions: reduced delay lowers the audio age, while the larger coefficient $\Psi_{d,2}$ scales up the weighted term in MAoI. Consistent with this mechanism, as the audio weight increases, the audio average AoI decreases, whereas the audio average MAoI exhibits the non-monotonic behavior described above. 

For the image and other signal modalities, both AoI and MAoI increase because their queueing times increase. 
At the system level, average MAoI shows an upward trend since the rise in non-audio MAoI dominates. 
The system average AoI also increases because the aggregate increase of non-audio ages outweighs the decrease achieved for audio modality. 

Overall, these results show that the modality weights provide effective, targeted control of freshness across modalities: a larger modality weight raises the priority of the corresponding modality and improves its timeliness, while incurring a modest increase in system-level AoI.

\subsection{Evaluation of the JSO}
To evaluate the effectiveness of the proposed JSO algorithm, we compare it with several representative baseline schemes under identical simulation settings. For each baseline except the Fixed Minimum Interval Algorithm (FMI), the sampling intervals and Lagrange multipliers are iteratively optimized as in the proposed JSO algorithm, with the offloading decisions determined by the corresponding scheme. The baseline schemes are summarized as follows:

\begin{itemize}

\item {\textbf{Fixed Minimum Interval (FMI):}} For each device, the sampling interval is fixed to the minimum energy feasible value. The offloading decisions and Lagrange multipliers are iteratively updated as in the JSO algorithm.

\item {\textbf{Full Local Computing (FLC):}} All devices always perform computations locally without offloading. 


\item {\textbf{Greedy Marginal-Cost Offloading (GMO):}} All devices start with local execution. In each iteration, GMO temporarily enables offloading for each undecided device and evaluates the system cost; it then fixes the device that yields the largest marginal decrease. Once a device is set to offload, the decision is not reversed.

\item {\textbf{Independent Distributed Decision (IDD):}} Each device decides on offloading using only local conditions, treating other signal devices with fixed average interference. 

\item {\textbf{Device-Wise Best-Response Offloading (DBRO):}} Each device, given the prior offloading decisions of other signals, compares local and edge costs and chooses its best response. Unlike JSO, which updates the offloading decision of the single device that yields the largest system-level cost decrease, DBRO performs sequential updates based on individual cost reductions.

\end{itemize}

\begin{figure}
\centering
\includegraphics{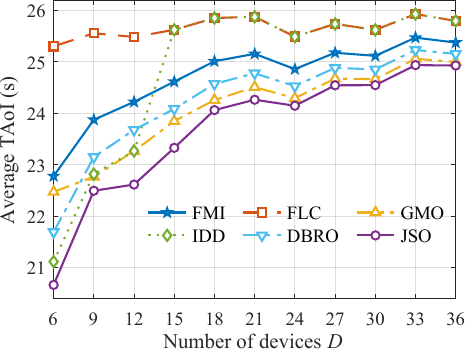}
\caption{Average MAoI vs. the number of devices $D$ for JSO and baselines.}
\label{fig:diff_device_algo}
\end{figure}
\begin{figure}
\centering
\includegraphics{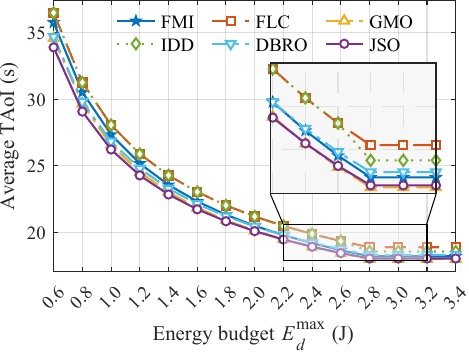}
\caption{Average MAoI vs. the energy budget $E_d^{\max}$ for {JSO} and baselines.}
\label{fig:diff_energy}
\end{figure}
\begin{figure}
\centering
\includegraphics{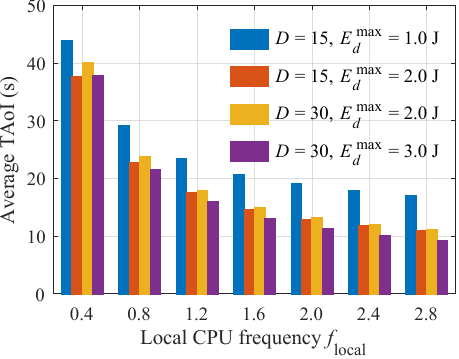}
\caption{System average MAoI vs. the local CPU frequency $f_{\text{local}}$.}
\label{fig:local_computing}
\end{figure}

Fig.~\ref{fig:diff_device_algo} shows the system average MAoI versus the number of devices $D$ for {JSO} and representative baselines. 
{JSO} consistently attains the lowest MAoI across all numbers of devices $D$, confirming the advantage of jointly adapting sampling and offloading. 
As $D$ increases, the average MAoI generally rises because mutual interference on the uplink limits how many devices can offload; consequently, only a limited subset of devices can offload, and a growing fraction must execute locally, which lengthens information cycles.
An exception is FLC, the system average MAoI is nearly flat with $D$ since all devices always compute locally and thus avoid the shared uplink.

Fig.~\ref{fig:diff_energy} plots the system average MAoI versus the per-device energy budget $E_d^{\max}$, comparing {JSO} with baselines. Across the evaluated range of $E_d^{\max}$, {JSO} achieves the lowest average MAoI. For all algorithms, average MAoI decreases as $E_d^{\max}$ increases because a larger energy budget allows shorter sampling intervals.
Once $E_d^{\max}$ exceeds 2.8~J, the system average MAoI shows little further reduction. In this range, devices already operate at the minimum sampling interval, so energy no longer constrains the sensing rate, and further increases in $E_d^{\max}$ do not lower MAoI.
At high energy budgets, {GMO} performs similarly to {JSO} because the sampling interval is fixed at its minimum, which leaves only offloading decisions and thus reduces {JSO} to the GMO.

Fig.~\ref{fig:local_computing} plots the system average MAoI versus the local CPU frequency $f_{\text{local}}$ across multiple settings of the number of devices $D$ and the maximum energy budget $E_d^{\max}$. In all cases, increasing $f_{\text{local}}$ reduces the system average MAoI because higher $f_{\text{local}}$ shortens local processing and queueing for all modalities.
The improvement diminishes at high $f_{\text{local}}$. As $f_{\text{local}}$ increases, local compute delay still decreases but at a slower rate. 
In this regime, MAoI is shaped mainly by the chosen sampling interval and offloading decisions, and additional increases in $f_{\text{local}}$ provide limited gains.

\begin{figure}
\centering
\includegraphics{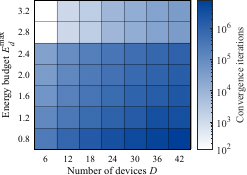}
\caption{Convergence iterations over a grid of number of  devices $D$ and energy budgets $E_d^{\max}$.}
\label{fig:hotmap_converge}
\end{figure}

Fig.~\ref{fig:hotmap_converge} reports the number of iterations to convergence over a grid of the number of devices $D$ and energy budgets $E_d^{\max}$ under the stopping rule described in the methodology. The iteration count increases with $D$ and decreases as $E_d^{\max}$ grows, and it remains within a practical range across the tested grid.
Fig.~\ref{fig:line_converge} further shows the scaling with $D$ at several energy budgets, with curves ordered by $E_d^{\max}$ such that larger budgets require fewer iterations for all $D$. When $E_d^{\max}$ is 3~J, the curves are nearly flat with respect to $D$. At this energy level, the sampling interval is fixed at its minimum, and repeated interval updates are unnecessary. Convergence is then driven mainly by offloading updates, which stabilize quickly.
\begin{figure}
\centering
\includegraphics{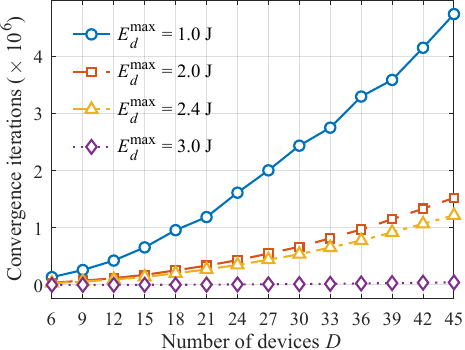}
\caption{Convergence iterations vs. the number of devices $D$ under several energy budgets $E_d^{\max}$.}
\label{fig:line_converge}
\end{figure}

\section{Conclusion}\label{sec:conclusion}
In this study, MAoI was introduced to quantify modality‑aware freshness for images, audio, and signals by linking age growth to content dynamics, semantic variation, and measurement quality. 
A closed-form expression of the average MAoI in a MEC system was derived, and an average-MAoI minimization problem was formulated. To solve the problem, a BCD-based algorithm named JSO is proposed, which alternates updates of the sampling intervals and the offloading decisions. 
Extensive simulations demonstrated the effectiveness of MAoI as a freshness metric. Compared with traditional AoI, MAoI more faithfully captured multimodal freshness by incorporating modality-specific semantic dynamics and temporal characteristics. Moreover, the proposed optimization framework consistently reduced system-level average MAoI relative to state-of-the-art baselines while satisfying per-device energy budgets. Additional studies confirmed stable convergence of the algorithm.
In summary, this work offered a unified framework for modality-aware AoI design, with potential implications for efficient multi-modal edge intelligence.

\end{document}

%% file: acronyms.tex
\newacronym{mec}{MEC}{millimeter-Wave}
\newacronym{taoi}{TAoI}{Modality-Tailored Age of Information}
\newacronym{aoi}{AoI}{Age of Information}
\newacronym{iot}{IoT}{Internet of Things}